\documentclass[aps,twocolumn,pra,showpacs,superscriptaddress]{revtex4-1}
\usepackage{amssymb}

\usepackage{amsmath}
\usepackage{graphicx,color}
\usepackage{epstopdf}
\usepackage{mathrsfs}
\usepackage{epsfig}
\usepackage[breaklinks=true]{hyperref}

\newcommand{\beq}{\begin{equation}}
\newcommand{\eeq}{\end{equation}}
\newcommand{\bqa}{\begin{eqnarray}}
\newcommand{\eqa}{\end{eqnarray}}

\definecolor{green}{rgb}{0.00,0.50,0.00}

\newtheorem{theorem}{Theorem}

\newtheorem{definition}[theorem]{Definition}

\newtheorem{proposition}[theorem]{Proposition}
\newtheorem{remark}[theorem]{Remark}

\newenvironment{proof}[1][Proof]{\noindent\textbf{#1.} }{\ \rule{0.5em}{0.5em}}

\begin{document}

%\maketitle

\title{Non-Markovian Quantum Feedback Networks I: Quantum Transmission Lines, Lossless Bounded Real Property and Limit Markovian Channels}

\author{John E.~Gough} \email{jug@aber.ac.uk}
   \affiliation{Aberystwyth University, SY23 3BZ, Wales, United Kingdom}

\begin{abstract}
The purpose of this paper is to set out the problems of modeling quantum communication and signal processing where the communication between systems via a non-Markovian channel. This is a general feature of quantum transmission lines. Our ultimate objective is to extend the networks rules that have been developed for Markovian models. To this end we recall the Hamiltonian description of such non-Markov models of transmission lines and their quantization. These have occurred in the context of non-equilibrium thermodynamics, but our interest is in the transmission lines as carriers of information rather than heat baths. We show that there is an analytic scattering matrix associated with these models and that stability may be formulated in terms of the lossless bounded real property.
Noting that the input and output fields do not separately satisfy a non-self-demolition principle, we discuss the rigorous limit in which such models appear Markov and so amenable to standard approaches of quantum filtering and control.
\end{abstract}

\maketitle

\section{Introduction}
\label{sec:introduction}
Markovianity is a standard modeling assumption for dealing with stochastic systems in the physical and engineering sciences. From it we obtain a simple probabilistic structure where all one requires is an initial state and a transition mechanism to propagate from present to future states, and from this one can compute multi-time expectations. Moreover, Markov models can be dilated - that is, seen as the sub-dynamics of a larger model described, for instance, by stochastic differential equations. In many applications, one starts with such a system+noise model. There exists an enormous literature on the filtering and control of such models. As is well-known, Markovianity is a property of convenience and is by no means a natural situation. For instance, a subsystem of a Markov system will typically not be Markov: in such cases, the complement of the subsystem may act as memory. Conversely, given a non-Markov model, one often tries to realize it as a sub-system of a Markov model. Central to the modeling procedure is the identification of the system so that there is sufficient information contained in its state to give the Markov property: ideally this is minimal, that is, we have no more degrees of freedom in the model than just the system's, however this is not essential.

In this article we will be concerned with quantum open systems however much of our discussion is of relevance to classical stochastic models. There has been a surge of interest in recent years in so-called \lq\lq non-Markovian quantum models". Typically these deal with a system coupled to an external environment (typically a infinite assembly of oscillators with spectral densities $\mathbf{J} (\omega )$) where the resulting dynamical evolution of the system exhibits memory effects. (We should stress that the term non-Markovian is typically used in this context by theoretical physicists. But, as we shall see in the first model of a quantum transmission line with ohmic spectral density below, there are natural non-Markovian models that actually exhibit a memory-less property.) The non-Markovian feature is then frequently presented as a resource - for example a \lq\lq quantum memory" - that may be exploited for quantum technologies. From the engineering viewpoint, described above, this is a rather unusual way to proceed, however this is explained by the fact that the starting point here is often a full physical model.

On the other hand, quantum technology has progressed in recent years to the stage where principles from (classical) engineering are now being proposed as essential control and stabilization techniques: these are formulated within the Markovian regime. An increasingly important role is being played by the concept of feedback. Measurement-based feedback was proposed by Belavkin \cite{Belavkin1}, and subsequently developed as a realistic tool in quantum optics by Wiseman and Milburn \cite{WM_book}. Central to this has been the used of dilations of quantum Markov systems using quantum stochastic calculus \cite{HP84}-\cite{Par92}. More recently, there has been considerable interest in autonomous schemes which avoid the computational overheads and delays, and the unwanted side effects of quantum measurement \cite{WM94} - \cite{Crisafulli13}. These belong to the category of coherent feedback quantum control schemes. With M.R. James, the author has developed a systematic approach to modeling networks of locally Markovian systems connected by instantaneous quantum field transmission lines \cite{GouJam09a,GouJam09b}. Here one can give explicit rules for the feedforward, feedback, cascading, super-positioning, etc., of arbitrary network architectures. (We remark that a object oriented visual hardware language for assembling devices, and then working with devices as basic objects has been developed for Markovian Quantum feedback networks by Tezak et al., \cite{Tezak}.) This modeling famework has lead to rapid development in quantum feedback engineering in recen years \cite{GGY}-\cite{GZ}.

However, the natural \lq\lq Hamiltonian'' description for open circuits is not Markovian (in the probabilistic sense) and generally not even memory-less either, see \cite{Louisell}-\cite{GarZol00}. The aim of this paper is to begin the program of extending the quantum feedback network theory to realistic models of system-field coupling taking standard models from electronics as guide. To some extent this has already been initiated in \cite{YD1984}, however, not to the extent of giving systematic rules for interconnection. We should also mention the recent analysis in \cite{ZLWJN13} which starts with the non-Markov models as treated in \cite{GarZol00} and sets about finding the analogue of our series product rule \cite{GouJam09a}. In electrical engineering one uses network theory \cite{Anderson_Vongpanitlerd} to describe composite systems, then these devices may be replaced by an equivalent circuit that has the same input-output ($I-V$) characteristics. The success of electrical engineering is that one can give simple rules for connecting components into circuits (the Kirchhoff current and voltage laws) and deduce powerful network theorems \cite{Belevitch}, \cite{Brune}. 

During the development of quantum models for open systems, there was a marked split amongst the quantum probability community between the Markovian and non-Markovian models. The Markovian regime allows for a remarkable extension of many of the core practical concepts of classical stochastic processes to the non-commutative setting, including technologically important methodologies such as filtering. It is now used extensively to model decoherence and dissipation in physical systems. The non-Markovian models emerged from relaxation to thermal equilibrium problems and the motivations often placed considerable emphasis on the requirements of physical correctness to ensure consistency with the the laws of thermodynamics. The reality is that both forms are widely used by theoreticians, though in different applications. In quantum optics, the auto-correlation time associated to a bath of photons is so small compared with the timescale for the system, and one would seriously question why a non-Markovian model is needed. It is here that the power of the quantum probabilistic program has been most apparent in current quantum technologies. The non-Markov models are essential however outside of optical models and it is likely that such approaches are required for treating superconducting qubit systems, solid state systems, etc. As such, we might reasonably expect that a developed theory of non-Markovian quantum feedback networks would be of use for circuit QED systems \cite{Devoret_2015}.

The outline of the paper is as follows. In Section \ref{Sec:Quant} we recall the theory of (ohmic) quantum transmission lines. In Section \ref{sec:Memory} we discuss the extension to models with memory and show that the input-output theory is described by an scattering matrix $\mathbf{S} [s]$ related to the Laplace transform of the memory kernel. As the systems of equations is linear we may apply the well-known description of dissipativity by Willems \cite{Willems}. It is previously known that the Laplace transform of the memory kernel must be positive real, in the sense of Cauer \cite{Cauer}, and this has been argued as an essential criterion for thermodynamic stability of a passive linear system \cite{Meixner} and for the current class of models when the input fields are in a thermal state \cite{FLC88}. Our main result of this section is Theorem \ref{thm:S_LBR} which establishes that the appropriate control-theoretic requirement on $\mathbf{S} [s]$ is the lossless bounded real property and show that this corresponds to the lossless positive real property for the Laplace transform of the memory kernel. From this we deduce the spectral densities $\mathbf{J} (\omega )$ for the external transmission lines. For completeness we give a Hamiltonian model for the transmission line based on a set of spectral densities $\mathbf{J} (\omega )$.

Finally, in Section \ref{sec_Markov} we discuss the Markov limit for these models. This is a crucially important point, both technically and conceptually. In non-Markovian models we find that the charge in the transmission line $q(t)$ at a terminal point may be naturally split as the sum $q_{ \mathrm{in} } (t) + q_{ \mathrm{out} } (t)$ associated with an incoming and an outgoing field. However, the inputs do not commute with themselves at different times, that is the commutator $[ q_{ \mathrm{in} } (t) ,
q_{ \mathrm{in} } (t^\prime )]$ is a non-trivial distribution $g (t-t^\prime )$, (similarly for $q_{ \mathrm{in} } (t)$). To coin a phrase based on \cite{Belavkin1}, this means that the input field is self-demolishing - as is the output field - and therefore there is no possibility to construct a quantum measurement theory (much less quantum filtering theory!) around such models. The situation here is that markovianity is a much more essential requirement in quantum filtering theory than in classical. We give a rigorous formulation of the Markov limit for the general quantum non-Markovian models under consideration based on the quantum stochastic limits \cite{AFL}, \cite{AGL} of van Hove type (weak coupling). The existence of a well defined Markov is crucial for the formulation of a quantum filtering theory \cite{Belavkin1}, see also \cite{CarBook93} and \cite{BouvanHJam07}.

\section{Quantizing Electric Communication Circuits}
\label{Sec:Quant}

\subsection{Hamiltonian Formulation Of Electric Networks}
We give a condensed review of the Hamiltonian formulation of transmission lines and their quantization. The main references are \cite{Louisell} and \cite{YD1984}. The problem has a long history however, going back to the Lamb model in 1900 \cite{Lamb1900}, and the is intimately related to the Ford-Kac-Mazur model \cite{FKM65} for relaxation of a brownian particle to thermal equilibrium (though we do not wish to specify to a thermal state for the transmission line!). This was further developed in \cite{Maassen82} where connections with the Thirring-Schwabl model were made.

\subsubsection{Lossless Circuits}

We consider a component consisting of coupled inductors and capacitors which
may be influence by several external voltages. For each terminal we
attribute a charge $q_{k}\left( t\right) $. We may then collect these
together as a vector 
\begin{eqnarray*}
\mathbf{q}\left( t\right) =\left[ 
\begin{array}{c}
q_{1}\left( t\right) \\ 
\vdots \\ 
q_{n}\left( t\right)
\end{array}
\right]
\end{eqnarray*}
with a similar vector $\mathbf{i}\left( t\right) =\frac{d}{dt}\mathbf{q}%
\left( t\right) $ for the terminal currents. Likewise let $\mathbf{v}\left( t\right) $ be the
vector of terminal voltages $v_{j}\left( t\right) $. The circuit equation is
then 
\begin{eqnarray}
\mathbf{L}\frac{d^{2}}{dt^{2}}\mathbf{q}\left( t\right) +\mathbf{Kq}\left(
t\right) =\mathbf{v}\left( t\right) ,  \label{eq:circuit_eq}
\end{eqnarray}
where $\mathbf{L}\geq 0$ is the inductance matrix (taken to be invertible)
and $\mathbf{K}\geq 0$ is the capacitor matrix. 
These equations come from the Lagrangian 
\begin{eqnarray*}
L\left( \mathbf{q},\mathbf{\dot{q}},t\right) &=&\frac{1}{2}\mathbf{\dot{q}}%
\left( t\right) ^{\top }\mathbf{L\dot{q}}\left( t\right) -\frac{1}{2}\mathbf{%
q}\left( t\right) ^{\top }\mathbf{Kq}\left( t\right) \nonumber\\
&&+\mathbf{q}\left(
t\right) ^{\top }\mathbf{v}\left( t\right) .
\end{eqnarray*}

The energy of the component system is then 
\begin{eqnarray*}
\mathscr{E}_{\mathrm{component}} ( t ) =\frac{1}{2}\mathbf{i} (
t ) ^\top \mathbf{Li} ( t) +\frac{1}{2}\mathbf{q} (t) ^\top \mathbf{Kq} ( t)
\end{eqnarray*}
and we have 
\begin{eqnarray}
\mathscr{\dot{E}}_{\mathrm{component}}\left( t\right) &=&\mathbf{i}\left(
t\right) ^{\top }\left\{ \mathbf{L}\frac{d^{2}\mathbf{q}\left( t\right) }{%
dt^{2}}+\mathbf{Kq}\left( t\right) \right\} \nonumber\\
&\equiv& \mathbf{i}\left( t\right)
^{\top }\mathbf{v}\left( t\right) .  \label{eq:power_component}
\end{eqnarray}

The conjugate variables are defined by 
\begin{eqnarray}
\mathbf{p}=\mathbf{L\dot{q}}\left( t\right)
\end{eqnarray}
and in this case the Hamiltonian is 
\begin{eqnarray}
H_{\mathrm{component}}=\frac{1}{2}\mathbf{p}^\top \mathbf{L}^{-1}\mathbf{p} + \frac{1}{2} \mathbf{q}\mathbf{K}\mathbf{q}
\label{eq:Ham_comp}
\end{eqnarray}
to which we can add the external driving term $H_{\mathrm{ext}}=-\mathbf{q}%
\left( t\right) ^{\top }\mathbf{v}\left( t\right) $.

\subsubsection{Semi-Infinite Transmission Lines}

The problem of including a non-zero resistance term in the circuit has a
long history. The model cannot be a closed Hamiltonian one, and instead one
must consider the circuit as an open system coupled to an infinite
environment.

Consider an infinite transmission line with $Q(z,t)$ denoting the charge at
position $z$ at time $t$. This can be modelled as a continuum of components. The current and voltage are given
respectively by 
\begin{eqnarray}
I(z,t)=\frac{\partial Q(z,t)}{\partial t},\quad V(z,t)=-\frac{1}{\kappa }%
\frac{\partial Q(z,t)}{\partial z}  \label{eq:def_I,V}
\end{eqnarray}
where $\kappa $ is the capacitance per unit length. We obtain immediately
that $\frac{\partial I}{\partial z}=-\kappa \frac{\partial V}{\partial t}$%
, while a second dynamical equation is 
\begin{eqnarray*}
\frac{\partial V}{\partial z}=-\ell \frac{\partial I}{\partial t}
\end{eqnarray*}
where we introduce the inductance $\ell $ per unit length. The two equations
together imply the wave equation 
\begin{eqnarray}
\frac{\partial ^{2}Q}{\partial z^{2}}-\frac{1}{c^{2}}\frac{\partial ^{2}Q%
}{\partial t^{2}}=0,  \label{eq:wave_eqr}
\end{eqnarray}
known as the \textit{telegrapher's equation} which is likewise satisfied by $%
I$ and $V$. The wave speed is given by $c=\frac{1}{\sqrt{\ell \kappa }}$.

The wave equation can be derived from the Lagrangian density 
\begin{eqnarray}
\mathscr{L}=\frac{1}{2}\ell \left( \frac{\partial Q}{\partial t}\right)
^{2}-\frac{1}{2\kappa }\left( \frac{\partial Q}{\partial z}\right) ^{2}
\end{eqnarray}
so that the Euler-Lagrange equations 
\begin{eqnarray*}
\frac{\partial }{\partial t}\left( \frac{\partial \mathscr{L}}{\partial
(\partial Q/\partial t)}\right) +\frac{\partial }{\partial z}\left( \frac{%
\partial \mathscr{L}}{\partial (\partial Q/\partial z)}\right) +\frac{%
\partial \mathscr{L}}{\partial Q}=0,
\end{eqnarray*}
yield (\ref{eq:wave_eqr}). The canonically conjugate field to $Q(z,t)$ is 
\begin{eqnarray}
\pi (z,t)=\frac{\partial \mathscr{L}}{\partial (\partial Q/\partial t)}%
=\ell \frac{\partial Q}{\partial t},
\end{eqnarray}
or $\pi =\ell I$, and from this we obtain the energy density 
\begin{eqnarray}
\mathscr{H}=\pi \frac{\partial Q}{\partial t}-\mathscr{L}=\frac{1}{2\ell }%
\pi ^{2}+\frac{1}{2\kappa }(\frac{\partial Q}{\partial z})^{2}.
\end{eqnarray}

We now consider $n$ transmission lines, each one connecting to one of the
terminals of the component.

The $j$th line has an associated charge distribution $Q_{j}\left(
z_{j},t\right) $ where $z_{j}\geq 0$ is the spatial coordinate along the $j$%
th line: we have $z_{j}=0$ at the $j$th terminal, and when no confusion
occurs we will often write just $z$ for the variables. The energy associated
with the transmission lines is therefore 
\begin{eqnarray*}
\mathscr{E}_{TL}\left( t\right) =\sum_{j}\int_{0}^{\infty }\left\{ \frac{%
\ell }{2}\left( \frac{\partial Q_{j}}{\partial t}\right) ^{2}+\frac{1}{%
2\kappa }\left( \frac{\partial Q_{j}}{\partial z}\right) ^{2}\right\} dz.
\end{eqnarray*}

\begin{proposition}
The transmission line energy changes at the rate 
\begin{eqnarray}
\mathscr{\dot{E}}_{TL}\left( t\right) \equiv \mathbf{I}^{\top }\left(
0,t\right) \mathbf{V}\left( 0,t\right) .  \label{eq:power_TL}
\end{eqnarray}
\end{proposition}

\begin{proof}
This follows from 
\begin{eqnarray*}
&&\frac{d}{dt}\sum_{j}\int_{0}^{\infty }\left\{ \frac{\ell }{2}\left( \frac{%
\partial Q_{j}}{\partial t}\right) ^{2}+\frac{1}{2\kappa }\left( \frac{%
\partial Q_{j}}{\partial z}\right) ^{2}\right\} dz \\
&=&\sum_{j}\int_{0}^{\infty }\left\{ \ell \frac{\partial Q_{j}}{\partial t}%
\frac{\partial ^{2}Q_{j}}{\partial t^{2}}+\frac{1}{\kappa }\frac{\partial
Q_{j}}{\partial z}\frac{\partial ^{2}Q_{j}}{\partial t\partial z}\right\} dz
\\
&=&\frac{1}{\kappa }\sum_{j}\int_{0}^{\infty }\left\{ \frac{\partial Q_{j}}{%
\partial t}\frac{\partial ^{2}Q_{j}}{\partial z^{2}}+\frac{\partial Q_{j}}{%
\partial z}\frac{\partial ^{2}Q_{j}}{\partial t\partial z}\right\} dz\quad 
\mathrm{(by (\ref{eq:wave_eqr}))} \\
&=&\frac{1}{\kappa }\sum_{j}\int_{0}^{\infty }\frac{\partial }{\partial z}%
\left\{ \frac{\partial Q_{j}}{\partial t}\frac{\partial Q_{j}}{\partial z}%
\right\} dz \\
&=&-\frac{1}{\kappa }\sum_{j}\left. \frac{\partial Q_{j}}{\partial t}%
\right| _{z=0}\left. \frac{\partial Q_{j}}{\partial z}\right| _{z=0} \\
&\equiv &\mathbf{I}^{\top }\left( 0,t\right) \mathbf{V}\left( 0,t\right) ,
\end{eqnarray*}
where we use (\ref{eq:def_I,V})).
\end{proof}

\subsubsection{Boundary Conditions}

The charge $q_{j}\left( t\right) $ at the $j$th terminal is taken to
coincide with the value of the charge at the origin of a semi-infinite
transmission line along the positive $z$-axis. In vector form, we have 
\begin{eqnarray}
\mathbf{q}(t)=\left. \mathbf{Q}(z,t)\right| _{z=0},
\label{eq:constraint}
\end{eqnarray}
and therefore 
\begin{eqnarray*}
\mathbf{i}(t)=\left. \frac{\partial }{\partial t}\mathbf{Q}(z,t)\right|
_{z=0}\equiv \mathbf{I}\left( 0,t\right) .
\end{eqnarray*}

The total energy will be 
\begin{eqnarray*}
\mathscr{E}\left( t\right) =\mathscr{E}_{\mathrm{component}}\left( t\right) +%
\mathscr{E}_{\mathrm{TL}}\left( t\right)
\end{eqnarray*}
and from (\ref{eq:power_component}) and (\ref{eq:power_TL}) we will have $%
\mathscr{E}\left( t\right) $ constant provided that 
\begin{eqnarray}
\mathbf{v}\left( t\right) \equiv -\mathbf{V}\left( 0,t\right) .
\label{eq:BC}
\end{eqnarray}
Using the circuit equation (\ref{eq:circuit_eq}), the boundary condition may
be written in terms of the fields as 
\begin{eqnarray}
\mathbf{L}\frac{\partial ^{2}\mathbf{Q}}{\partial t^{2}}\left( 0,t\right) +%
\mathbf{KQ}\left( 0,t\right) =\frac{1}{\kappa }\frac{\partial \mathbf{Q}}{%
\partial z}\left( 0,t\right) .  \label{eq:BC_Q}
\end{eqnarray}

A mechanical analogue dates back to Lamb in 1900 \cite{Lamb1900}. Here the lossy nature of the problem is in clear evidence - the energy is free to escape to infinity along the string.

\subsubsection{The Circuit Equation}

The charge densities $Q_{j}(z,t)$ satisfy the wave equation for $z>0$, we
again have the general solution (\ref{eq:wave_eqr}) which we now rewrite as 
\begin{eqnarray*}
\mathbf{Q}(z,t)=\mathbf{q}_{\mathrm{in}}(t+\frac{z}{c})+\mathbf{q}_{\mathrm{%
out}}(t-\frac{z}{c})
\end{eqnarray*}
with the $z>0$ left and right propagating parts interpreted as inputs and
outputs. In particular, 
\begin{eqnarray}
\mathbf{q}(t) = \mathbf{q}_{\mathrm{in}}(t)+\mathbf{q}_{\mathrm{out}}(t)
\label{eq:q_in_out}
\end{eqnarray}
We see that 
\begin{eqnarray*}
\mathbf{I}(z,t) &=&\frac{\partial \mathbf{Q}}{\partial t}=\mathbf{\dot{q}}_{%
\mathrm{in}}(t+\frac{z}{c})+\mathbf{\dot{q}}_{\mathrm{out}}(t-\frac{z}{c}),
\\
\mathbf{V}(z,t) &=&-\frac{1}{\kappa }\frac{\partial \mathbf{Q}}{\partial z}%
=-\frac{1}{\kappa c}\mathbf{\dot{q}}_{\mathrm{in}}(t+\frac{z}{c})+\frac{1}{%
\kappa c}\mathbf{\dot{q}}_{\mathrm{out}}(t-\frac{z}{c}),
\end{eqnarray*}
and we may eliminate $\mathbf{\dot{q}}_{\mathrm{out}}$ to get 
\begin{eqnarray*}
\mathbf{V}(z,t)=\frac{1}{\kappa c}\mathbf{I}(z,t)-\frac{2}{\kappa c}%
\mathbf{\dot{q}}_{\mathrm{in}}(t+\frac{z}{c})
\end{eqnarray*}
and so the terminal voltages $\mathbf{v}(t)=\mathbf{V}(z=0,t)$ are 
\begin{eqnarray*}
\mathbf{v}(t) &=&-\frac{1}{\kappa c}\mathbf{I}(0,t)+\frac{2}{\kappa c}%
\mathbf{\dot{q}}_{\mathrm{in}}(t) \\
&\equiv &-R\mathbf{\dot{q}}(t)+\mathbf{F}(t),
\end{eqnarray*}
where $R=\frac{1}{\kappa c}$ and we have the \lq\lq noise \rq\rq
\begin{eqnarray*}
\mathbf{F}(t)=2R\,\mathbf{\dot{q}}_{\mathrm{in}}(t).
\end{eqnarray*}
We note that the new constant $R$ is in fact $\sqrt{\frac{\ell }{\kappa }}%
=Z $, the characteristic impedance of the transmission lines. (We could obviously
generalize to different characteristics for each line.)

The equations of motion for the terminal then take the remarkable form of a Langevin equation 
\begin{eqnarray}
\mathbf{L\ddot{q}}(t)+R\mathbf{\dot{q}}(t)+\mathbf{Kq}(t)=\mathbf{F}(t),
\label{eq:Langevin}
\end{eqnarray}
which are of the form of a normal electric circuit with inductance matrix $\mathbf{L}$, capacitance matrix $\mathbf{K}$, along
with resistors $R$ and an incoming driving voltage $\mathbf{F}(t)$.

We note that 
\begin{eqnarray*}
\frac{d}{dt}\left[ 
\begin{array}{c}
\mathbf{q}(t) \\ 
\mathbf{i}(t)
\end{array}
\right] =\mathbf{A}_{-}\left[ 
\begin{array}{c}
\mathbf{q}(t) \\ 
\mathbf{i}(t)
\end{array}
\right] +2R\left[ 
\begin{array}{c}
0 \\ 
\mathbf{L}^{-1}\mathbf{\dot{q}}_{\mathrm{in}}(t)
\end{array}
\right]
\end{eqnarray*}
where 
\begin{eqnarray*}
\mathbf{A}_{-}=\left[ 
\begin{array}{cc}
0 & I_{n} \\ 
-\mathbf{L^{-1}K} & -R\mathbf{L}^{-1}
\end{array}
\right]
\end{eqnarray*}
which we assume to be Hurwitz.

We may then may integrate to get 
\begin{eqnarray}  \label{eq:integrate}
\left[ 
\begin{array}{c}
\mathbf{q}(t) \\ 
\mathbf{i}(t)
\end{array}
\right] &=&e^{\mathbf{A}_- t}\left[ 
\begin{array}{c}
\mathbf{q}(0) \\ 
\mathbf{i}(0)
\end{array}
\right] \nonumber \\
&&+2R\int_{0}^{t}e^{\mathbf{A}_- (t-\tau )}\left[ 
\begin{array}{c}
0 \\ 
\mathbf{L}^{-1}\mathbf{\dot{q}}_{\mathrm{in}}(\tau )
\end{array}
\right] d\tau  \nonumber \\
\end{eqnarray}
which implies an expression of the form 
\begin{eqnarray*}
\mathbf{i}(t) &\equiv &\mathbf{H}_{iq}^{-}\left( t\right) \mathbf{q}(0)+\mathbf{%
H}_{ii}^{-}\left( t\right) \mathbf{i}(0)\\
&& +\int_{0}^{t}\mathbf{G}^{-}\left(
t-\tau \right) \mathbf{\dot{q}}_{\mathrm{in}}(\tau ),
\end{eqnarray*}
where 
\begin{eqnarray*}
\left[ 
\begin{array}{cc}
\mathbf{H}_{qq}^{-}\left( t\right) & \mathbf{H}_{qi}^{-}\left( t\right) \\ 
\mathbf{H}_{iq}^{-}\left( t\right) & \mathbf{H}_{ii}^{-}\left( t\right)
\end{array}
\right] =e^{\mathbf{A}^{-}t}
\end{eqnarray*}
and 
\begin{eqnarray*}
\mathbf{G}^{-}\left( t\right) =2R\,\mathbf{H}_{ii}^{-}\left( t\right) 
\mathbf{L}^{-1}.
\end{eqnarray*}

If we wished to solve for negative times, then we encounter the analogue expressions
\begin{eqnarray*}
\mathbf{A}_{+}=\left[ 
\begin{array}{cc}
0 & I_{n} \\ 
-\mathbf{L^{-1}K} & +R\mathbf{L}^{-1}
\end{array}
\right] ,
\end{eqnarray*}
so that $-\mathbf{A}_{+}$ is Hurwitz, and 
\begin{eqnarray*}
\left[ 
\begin{array}{cc}
\mathbf{H}_{qq}^{+}\left( t\right) & \mathbf{H}_{qi}^{+}\left( t\right) \\ 
\mathbf{H}_{iq}^{+}\left( t\right) & \mathbf{H}_{ii}^{+}\left( t\right)
\end{array}
\right] =e^{\mathbf{A}_{+}t}
\end{eqnarray*}
and 
\begin{eqnarray*}
\mathbf{G}^{+}\left( t\right) =2R\,\mathbf{H}_{ii}^{+}\left( t\right) 
\mathbf{L}^{-1}.
\end{eqnarray*}

\subsubsection{Initial Data, Input and Output Spaces}

It is convenient to introduce the parameter $\tau =z/c$ along transmission
lines which measures the time take for a signal to travel in from a point on
the line to the corresponding terminal. Therefore $\mathbf{Q}(z,t)\equiv 
\mathbf{q}_{\mathrm{in}}(t+\tau )+\mathbf{q}_{\mathrm{out}}(t-\tau )$. Let
us introduce the \textit{time-zero field data} 
\begin{eqnarray*}
f_{j}\left( \tau \right) &=&\left. Q_{j}\left( t,c\tau \right) \right|
_{t=0}, \\
g_{j}\left( \tau \right) &=&\left. \frac{\partial Q_{j}\left( t,c\tau
\right) }{\partial t}\right| _{t=0}.
\end{eqnarray*}
(Note that these functions are only defined for $\tau \geq 0$.) We then have 
\begin{eqnarray*}
\mathbf{f}\left( \tau \right) &=&\mathbf{q}_{\mathrm{in}}\left( \tau \right)
+\mathbf{q}_{\mathrm{out}}\left( -\tau \right) , \\
\mathbf{\dot{f}}\left( \tau \right) &=&\mathbf{\dot{q}}_{\mathrm{in}}\left(
\tau \right) -\mathbf{\dot{q}}_{\mathrm{out}}\left( -\tau \right) , \\
\mathbf{g}\left( \tau \right) &=&\mathbf{\dot{q}}_{\mathrm{in}}\left( \tau
\right) +\mathbf{\dot{q}}_{\mathrm{out}}\left( -\tau \right) .
\end{eqnarray*}
and rearranging 
\begin{eqnarray}
\mathbf{\dot{q}}_{\mathrm{in}}\left( \tau \right) &=&\frac{1}{2}\mathbf{g}%
\left( \tau \right) +\frac{1}{2}\mathbf{\dot{f}}\left( \tau \right) ,\quad
\left( \tau >0\right) ;  \label{eq:dotqin_fg} \\
\mathbf{\dot{q}}_{\mathrm{out}}\left( \tau \right) &=&\frac{1}{2}\mathbf{g}%
\left( -\tau \right) -\frac{1}{2}\mathbf{\dot{f}}\left( -\tau \right) ,\quad
\left( \tau <0\right) .  \label{eq:dotqout_fg}
\end{eqnarray}

\begin{proposition}
\label{Prop:dotq_out_fg}The output field $\mathbf{\dot{q}}_{\mathrm{out}%
}\left( t\right)$  is given in terms of the
time-zero data fields by  ($-\infty <t<\infty $)
\begin{eqnarray}
\mathbf{\dot{q}}_{\mathrm{out}}\left( t\right) &=&\left\{ 
\begin{array}{l}
\mathbf{H}_{iq}^{-}\left( t\right) \mathbf{f}(0)+\mathbf{H}_{ii}^{-}\left(
t\right) \mathbf{g}(0) \\
+\frac{1}{2}\int_{0}^{t}\mathbf{G}^{-}\left( t-\tau \right) \left( \mathbf{g}%
\left( \tau \right) +\mathbf{\dot{f}}\left( \tau \right) \right) d\tau \\
-\frac{1}{2}\mathbf{g}\left( t\right) -\frac{1}{2}\mathbf{\dot{f}}(t)
, \qquad t\geq 0; \\ 
\\
\frac{1}{2}\mathbf{g}\left( -t\right) -\frac{1}{2}\mathbf{\dot{f}}\left(
-t\right) , \qquad t\leq 0.
\end{array}
\right.  \nonumber \\
&&  \label{eq:fg_to_out}
\end{eqnarray}
and similarly
\begin{eqnarray}
\mathbf{\dot{q}}_{\mathrm{in}}\left( t\right) &=&\left\{ 
\begin{array}{l}
\frac{1}{2}\mathbf{g}\left( t\right) +\frac{1}{2}\mathbf{\dot{f}}\left(
t\right) , \qquad t\geq 0; \\ 
\\
\mathbf{H}_{iq}^{+}\left( t\right) \mathbf{f}(0)+\mathbf{H}_{ii}^{+}\left(
t\right) \mathbf{g}(0)  \\
+\frac{1}{2}\int_{0}^{-t}\mathbf{G}^{+}\left( t+\sigma \right) \left( 
\mathbf{g}\left( \sigma \right) +\mathbf{\dot{f}}\left( \sigma \right)
\right) d\sigma \\
-\frac{1}{2}\mathbf{g}\left( -t\right) +\frac{1}{2}\mathbf{\dot{f}}(-t)
, \qquad t\leq 0.
\end{array}
\right.  \nonumber \\
&&  \label{eq:fg_to_in}
\end{eqnarray}
\end{proposition}

\begin{proof}
From the identity (\ref{eq:q_in_out}) and the expression (\ref{eq:integrate}%
) we have for $t>0$%
\begin{eqnarray*}
\mathbf{\dot{q}}_{\mathrm{out}}\left( t\right) &=&\mathbf{i}(t)-\mathbf{\dot{%
q}}_{\mathrm{in}}(t) \\
&=&\mathbf{H}_{iq}^{-}\left( t\right) \mathbf{q}(0)+\mathbf{H}%
_{ii}^{-}\left( t\right) \mathbf{i}(0) \\
&&+\int_{0}^{t}\mathbf{G}^{-}\left( t-\tau \right) \mathbf{\dot{q}}_{\mathrm{%
in}}(\tau )d\tau -\mathbf{\dot{q}}_{\mathrm{in}}(t) \\
&\equiv &\mathbf{H}_{iq}^{-}\left( t\right) \mathbf{f}(0)+\mathbf{H}%
_{ii}^{-}\left( t\right) \mathbf{g}(0) \\
&&+\frac{1}{2}\int_{0}^{t}\mathbf{G}^{-}\left( t-\tau \right) \left( \mathbf{%
g}\left( \tau \right) +\mathbf{\dot{f}}\left( \tau \right) \right) d\tau \\
&&-\frac{1}{2}\mathbf{g}\left( t\right) -\frac{1}{2}\mathbf{\dot{f}}(t).
\end{eqnarray*}
and observing that $\mathbf{q}(0)=\mathbf{f}(0)$, $\mathbf{i}(0)=\mathbf{g}%
(0)$, and substituting in for $\mathbf{\dot{q}}_{\mathrm{in}}(\tau )$ in
terms of $\mathbf{g}\left( \tau \right) $ and $\mathbf{\dot{f}}\left( \tau
\right) $ using (\ref{eq:dotqin_fg}) yields the form for $t\geq 0$. For $t<0$%
, we have (\ref{eq:dotqout_fg}).

A similar argument leads to input $\mathbf{\dot{q}}_{\mathrm{in}}$ in terms
of the initial field data. 
\end{proof}

\bigskip

The energy in the transmission line may be rewritten in terms of the initial
field data as 
\begin{eqnarray*}
\mathscr{E}_{TL}\left( 0\right) =\frac{1}{2\kappa c}\sum_{j}\int_{0}^{%
\infty }\left\{ g_{j}\left( \tau \right) +\dot{f}_{j}\left( \tau \right)
^{2}\right\} d\tau
\end{eqnarray*}
and similarly the energy stored in the circuit is 
\begin{eqnarray*}
\mathscr{E}_{\mathrm{component}}\left( 0\right) =\frac{1}{2}\mathbf{g}\left(
0\right) ^{\top }\mathbf{Lg}\left( 0\right) +\frac{1}{2}\mathbf{f}\left(
0\right) ^{\top }\mathbf{Kf}\left( 0\right) .
\end{eqnarray*}

\begin{definition}
Consider the space of $n$ pairs of smooth compact-supported functions $%
\left( \mathbf{f},\mathbf{g}\right) $ on the domain $[0,\infty )$ with
``energy norm'' 
\begin{eqnarray*}
\mathscr{E}\left( \mathbf{f},\mathbf{g}\right) &=&\frac{1}{2}\mathbf{g}%
\left( 0\right) ^{\top }\mathbf{Lg}\left( 0\right) +\frac{1}{2}\mathbf{f}%
\left( 0\right) ^{\top }\mathbf{Kf}\left( 0\right) \\
&&+\frac{1}{2\kappa c}\int_{0}^{\infty }\left\{ \mathbf{g}\left( \tau
\right) ^{\top }\mathbf{g}\left( \tau \right) +\mathbf{\dot{f}}\left( \tau
\right) ^{\top }\mathbf{\dot{f}}\left( \tau \right) \right\} d\tau .
\end{eqnarray*}
The Hilbert space completion of the space with respect to the energy norm is
referred to as the initial data Hilbert space $\mathfrak{H}_{0}$. The energy
norm will then define the norm $\left\| \left( \mathbf{f},\mathbf{g}\right)
\right\| _{\mathscr{E}}^{2}=\kappa c\mathscr{E}\left( \mathbf{f},\mathbf{g}%
\right) $ (and the inner product!) in the obvious way.
\end{definition}

\begin{proposition}
\label{Proposition:energy_0_in_out}The total energy $\mathscr{E}=\mathscr{E}%
_{\mathrm{component}}+\mathscr{E}_{\mathrm{TL}}$ is given by 
\begin{eqnarray*}
\mathscr{E} &=&\frac{1}{\kappa c}\int_{-\infty }^{\infty }\mathbf{\dot{q}}_{%
\mathrm{out}}\left( \tau \right) ^{\top }\mathbf{\dot{q}}_{\mathrm{out}%
}\left( \tau \right) d\tau \\
&=&\frac{1}{\kappa c}\int_{-\infty }^{\infty }\mathbf{\dot{q}}_{\mathrm{in}%
}\left( \tau \right) ^{\top }\mathbf{\dot{q}}_{\mathrm{in}}\left( \tau
\right) d\tau .
\end{eqnarray*}
\end{proposition}

\begin{proof}
As $\mathscr{E}$ is constant it suffices to look at the time zero expression 
\begin{eqnarray*}
\mathscr{E} &=&\frac{1}{2}\left( \mathbf{\dot{q}}_{\mathrm{in}}\left(
0\right) +\mathbf{\dot{q}}_{\mathrm{out}}\left( 0\right) \right) ^{\top }%
\mathbf{L}\left( \mathbf{\dot{q}}_{\mathrm{in}}\left( 0\right) +\mathbf{\dot{%
q}}_{\mathrm{out}}\left( 0\right) \right) \\
&&+\frac{1}{2}\left( \mathbf{q}_{\mathrm{in}}\left( 0\right) +\mathbf{q}_{%
\mathrm{out}}\left( 0\right) \right) ^{\top }\mathbf{K}\left( \mathbf{q}_{%
\mathrm{in}}\left( 0\right) +\mathbf{q}_{\mathrm{out}}\left( 0\right) \right)
\\
&&+\sum_{j}\int_{0}^{\infty }\bigg\{\frac{\ell }{2}\left( \dot{q}_{\mathrm{in%
},j}\left( \frac{z}{c}\right) +\dot{q}_{\mathrm{out},j}\left( -\frac{z}{c}%
\right) \right) ^{2} \\
&&+\frac{1}{2\kappa c^{2}}\left( \dot{q}_{\mathrm{in},j}\left( \frac{z}{c}%
\right) -\dot{q}_{\mathrm{out},j}\left( -\frac{z}{c}\right) \right) ^{2}%
\bigg\}dz
\end{eqnarray*}
and the integral term simplifies to 
\begin{eqnarray*}
\frac{\ell c}{2}\sum_{j}\int_{0}^{\infty }\dot{q}_{\mathrm{in},j}\left( \tau
\right) ^{2}d\tau +\frac{\ell c}{2}\sum_{j}\int_{-\infty }^{0}\dot{q}_{%
\mathrm{out},j}\left( \tau \right) ^{2}d\tau .
\end{eqnarray*}
Now the boundary condition (\ref{eq:BC_Q}) is 
\begin{eqnarray*}
&&\mathbf{L}\left( \mathbf{\ddot{q}}_{\mathrm{in}}\left( t\right) +\mathbf{%
\ddot{q}}_{\mathrm{out}}\left( t\right) \right) +\mathbf{K}\left( \mathbf{q}%
_{\mathrm{in}}\left( t\right) +\mathbf{q}_{\mathrm{out}}\left( t\right)
\right) \\
&=&\frac{1}{\kappa c}\left( \mathbf{\dot{q}}_{\mathrm{in}}\left( t\right) -%
\mathbf{\dot{q}}_{\mathrm{out}}\left( t\right) \right)
\end{eqnarray*}

This implies 
\begin{eqnarray*}
&&\frac{d}{dt}\bigg\{\frac{1}{2}\left( \mathbf{\dot{q}}_{\mathrm{in}}\left(
t\right) +\mathbf{\dot{q}}_{\mathrm{out}}\left( t\right) \right) ^{\top }%
\mathbf{L}\left( \mathbf{\dot{q}}_{\mathrm{in}}\left( t\right) +\mathbf{\dot{%
q}}_{\mathrm{out}}\left( t\right) \right) \\
&&+\frac{1}{2}\left( \mathbf{q}_{\mathrm{in}}\left( t\right) +\mathbf{q}_{%
\mathrm{out}}\left( t\right) \right) ^{\top }\mathbf{K}\left( \mathbf{q}_{%
\mathrm{in}}\left( t\right) +\mathbf{q}_{\mathrm{out}}\left( t\right)
\right) \bigg\} \\
&=&\frac{1}{\kappa c}\left( \mathbf{\dot{q}}_{\mathrm{in}}\left( t\right) +%
\mathbf{\dot{q}}_{\mathrm{out}}\left( t\right) \right) ^{\top }\left( 
\mathbf{\dot{q}}_{\mathrm{in}}\left( t\right) -\mathbf{\dot{q}}_{\mathrm{out}%
}\left( t\right) \right)
\end{eqnarray*}
and integrating over the range $0\leq t<\infty $ and using the integration
by parts formula yields 
\begin{eqnarray*}
&&-\frac{1}{2}\left( \mathbf{\dot{q}}_{\mathrm{in}}\left( 0\right) +\mathbf{%
\dot{q}}_{\mathrm{out}}\left( 0\right) \right) ^{\top }\mathbf{L}\left( 
\mathbf{\dot{q}}_{\mathrm{in}}\left( 0\right) +\mathbf{\dot{q}}_{\mathrm{out}%
}\left( 0\right) \right) \\
&&-\frac{1}{2}\left( \mathbf{q}_{\mathrm{in}}\left( t\right) +\mathbf{q}_{%
\mathrm{out}}\left( 0\right) \right) ^{\top }\mathbf{K}\left( \mathbf{q}_{%
\mathrm{in}}\left( 0\right) +\mathbf{q}_{\mathrm{out}}\left( 0\right) \right)
\\
&=&\frac{1}{\kappa c}\sum_{j}\int_{0}^{\infty }\dot{q}_{\mathrm{in},j}\left(
\tau \right) ^{2}d\tau -\frac{1}{\kappa c}\sum_{j}\int_{0}^{\infty }\dot{q}_{%
\mathrm{out},j}\left( \tau \right) ^{2}d\tau
\end{eqnarray*}
and substituting in gives 
\begin{eqnarray*}
\mathscr{E} &=& \sum_{j} \left( \frac{1}{\kappa c} \int_{0}^{\infty }
+ \ell c \int_{-\infty}^{0} \right)
\dot{q}_{\mathrm{out},j} \left( \tau \right) ^{2}d\tau  \\
&=&\frac{1}{\kappa c}\int_{-\infty }^{\infty }\mathbf{\dot{q}}_{\mathrm{out}%
}\left( \tau \right) ^{\top }\mathbf{\dot{q}}_{\mathrm{out}}\left( \tau
\right) d\tau .
\end{eqnarray*}
The second form, $\mathscr{E}=\frac{1}{\kappa c}\int_{-\infty }^{\infty }%
\mathbf{q}_{\mathrm{in}}\left( \tau \right) ^{\top }\mathbf{q}_{\mathrm{in}%
}\left( \tau \right) d\tau $ is similarly established.
\end{proof}

\subsubsection{Input and Output Hilbert Spaces}

Let us define the \textit{Output Hilbert Space} to be Hilbert space $%
\mathfrak{H}_{\mathrm{out}}$\ of $n$ real-valued functions $\mathbf{\dot{q}}%
_{\mathrm{out}}$ on $\left( -\infty ,\infty \right) $ with norm $\left\| 
\mathbf{\dot{q}}_{\mathrm{out}}\right\| _{\mathrm{out}}^{2}=\int_{-\infty
}^{\infty }\mathbf{\dot{q}}_{\mathrm{out}}\left( \tau \right) ^{\top }%
\mathbf{\dot{q}}_{\mathrm{out}}\left( \tau \right) d\tau $. The mapping in
Proposition \ref{Prop:dotq_out_fg} given by (\ref{eq:fg_to_out}) is
therefore a mapping 
\begin{eqnarray*}
\mathscr{V}_{\mathrm{out}} :\mathfrak{H}_{0}\mapsto \mathfrak{H}_{\mathrm{out}} 
:\left( \mathbf{f},\mathbf{g}\right) \mapsto \mathbf{\dot{q}}_{\mathrm{out}}
\end{eqnarray*}
and by Proposition \ref{Proposition:energy_0_in_out} $\mathscr{V}_{\mathrm{%
out}}$ is unitary. 

Likewise we can define an \textit{Input Hilbert Space} $%
\mathfrak{H}_{\mathrm{in}}$\ of $n$ real-valued functions $\mathbf{\dot{q}}_{%
\mathrm{in}}$ on $\left( -\infty ,\infty \right) $ with norm $\left\| 
\mathbf{\dot{q}}_{\mathrm{in}}\right\| _{\mathrm{in}}^{2}=\int_{-\infty
}^{\infty }\mathbf{\dot{q}}_{\mathrm{in}}\left( \tau \right) ^{\top }\mathbf{%
\dot{q}}_{\mathrm{in}}\left( \tau \right) d\tau $, so that (\ref{eq:fg_to_in}%
) determines a unitary 
\begin{eqnarray*}
\mathscr{V}_{\mathrm{in}} :\mathfrak{H}_{0}\mapsto \mathfrak{H}_{\mathrm{in%
}} :\left( \mathbf{f},\mathbf{g}\right) \mapsto \mathbf{\dot{q}}_{\mathrm{in}}
\end{eqnarray*}
Combining the two gives the unitary $\mathscr{S}:\mathfrak{H}_{\mathrm{in}%
}\mapsto \mathfrak{H}_{\mathrm{out}}$%
\begin{eqnarray*}
\mathscr{S}=\mathscr{V}_{\mathrm{in}}^{\ast }\mathscr{V}_{\mathrm{out}}.
\end{eqnarray*}

\subsubsection{The Input-Output Relations}

We may write the Langevin equation (\ref{eq:Langevin}) in the frequency
domain as 
\begin{eqnarray*}
\mathbf{\hat{q}}(\omega )=\mathbf{\alpha }_{-}(\omega )\mathbf{\hat{F}}%
(\omega )
\end{eqnarray*}
where we introduce the susceptibility of the circuit as 
\begin{eqnarray*}
\mathbf{\alpha }_{-}(\omega )=\left[ -\mathbf{L}\omega ^{2}-iR\omega +%
\mathbf{K}\right] ^{-1}
\end{eqnarray*}

\begin{remark}
In the case of a single $LCR$ circuit with $\omega _{0}=\frac{1}{\sqrt{LC}}$
the natural frequency of the circuit, and $\gamma =\frac{R}{L}$ the circuit
damping constant, we have the scalar susceptibility $\alpha \left( \omega
\right) =\frac{1}{L}\frac{1}{\left( \omega _{0}^{2}-\omega ^{2}\right)
-i\gamma \omega }$. The poles of the susceptibility are at $\omega _{\pm }=i%
\frac{1}{2}\gamma \pm \sqrt{\omega _{0}^{2}-\frac{1}{4}\gamma ^{2}}$ and
we note that we always have damping as $\mathrm{Im}\,\omega _{\pm }<0$.
\end{remark}

Starting from the constraint relation (\ref{eq:q_in_out}) we have $\mathbf{%
\hat{q}}_{\mathrm{out}}(\omega )=\mathbf{\hat{q}}(\omega )-\mathbf{\hat{q}}_{%
\mathrm{in}}(\omega )$ or
\begin{eqnarray}
\mathbf{\hat{q}}_{\mathrm{out}}(\omega )=\mathbf{G}(\omega )\,\mathbf{\hat{q}%
}_{\mathrm{in}}(\omega ),\qquad \left( \omega \in \mathbb{R}\right) 
\label{eq:Q_scatter}
\end{eqnarray}
where the scattering matrix is given by
\begin{eqnarray}
\mathbf{G}(\omega )&=&2i\omega R\mathbf{\alpha }_-(\omega )-1\nonumber \\
&=& - \left( -%
\mathbf{L}\omega ^{2}+iR\omega +\mathbf{K}\right)  \left( -\mathbf{L}\omega
^{2}-iR\omega +\mathbf{K}\right)^{-1 }\nonumber \\
&=&- \mathbf{\alpha }_{-}(\omega ) \, \mathbf{\alpha }_{+}(\omega )^{-1}
\end{eqnarray}
where $\mathbf{\alpha }_{+}(\omega )=\mathbf{\alpha }_{-}(\omega )^{\ast }$.
In particular, $\mathbf{G}(\omega )$ is unitary for all real values $\omega $%
.

We may represent the functions $\mathbf{q}_{\mathrm{in}}$ and $\mathbf{q}_{%
\mathrm{out}}$ as 

\begin{eqnarray}
\mathbf{q}_{\mathrm{in}}(t ) \hspace{-0.2cm} &\equiv& \hspace{-0.3cm}
\int_{0}^{\infty } 
\sqrt{\frac{\hbar }{4\pi Z \omega}}
\left[ \mathbf{a}_{\mathrm{in}%
}(\omega )e^{-i\omega t }+\mathbf{a}_{\mathrm{in}}(\omega )^{\ast
}e^{i\omega t }\right] d\omega , \nonumber\\
\mathbf{q}_{\mathrm{out}}(t ) \hspace{-0.2cm} &\equiv  &  \hspace{-0.3cm}
\int_{0}^{\infty } 
\sqrt{\frac{\hbar }{4\pi Z \omega}}
\left[ \mathbf{a}_{\mathrm{out}%
}(\omega )e^{-i\omega t }+\mathbf{a}_{\mathrm{out}}(\omega )^{\ast
}e^{i\omega t }\right] d\omega , \nonumber\\
 \label{eq:free a}
\end{eqnarray}
with characteristic impedance $Z=\sqrt{\frac{\ell }{\kappa }}$ ohms. The
input-output relation (\ref{eq:Q_scatter}) therefore implies that
\begin{eqnarray}
\mathbf{a}_{\mathrm{out}}(\omega )=\mathbf{G}(\omega )\,\mathbf{a}_{\mathrm{%
in}}(\omega ),\qquad \left( \omega >0\right) .
\end{eqnarray}
It follows that
\begin{eqnarray*}
\mathscr{E} &=&\frac{1}{\kappa c}\int_{-\infty }^{\infty }\mathbf{\hat{\dot{q}}}_{\mathrm{out}}\left( \omega \right) ^{\top }
\mathbf{\hat{\dot{q}}}_{\mathrm{out}}  \left( \omega \right) d\omega \\
&=&\frac{1}{\kappa c}\int_{-\infty }^{\infty } \mathbf{\hat{\dot{q}}}_{\mathrm{in}}
\left( \omega \right) ^{\top } \mathbf{\hat{\dot{q}}}_{\mathrm{in}} \left( \omega
\right) d\omega ,
\end{eqnarray*}
which is in agreement with Proposition \ref{Proposition:energy_0_in_out} by the Plancherel theorem.

\subsection{Quantization}
The lossless circuit model is easily quantized by imposing the canonical commutation relations 
\begin{eqnarray}
\left[ q_{j},p_{k}\right] =i\hbar \;\delta _{jk}.
\label{eq:CCR_qp}
\end{eqnarray}
In particular, the Hamiltonian (\ref{eq:Ham_comp}) can now be interpreted as an operator and written in the form
\begin{eqnarray*}
H_{\mathrm{component}}\equiv \sum_{k}\hbar \omega _{k}\left( a_{k}^{\ast
}a_{k}+\frac{1}{2}\right)
\end{eqnarray*}
where we introduce normal modes satisfying $\left[ a_{j},a_{k}^{\ast }\right]
=\delta _{jk}$. 
\begin{remark}
In the $n=1$ case we have an inductance $L$ henries ($\texttt{Ns}^{2}\,
\texttt{C}^{-2}$) and a capacitance $C$ farads ($\texttt{C}^{2} \texttt{N}^{-1}$), but no resistance (%
$R=0$). Introducing the impedance $Z_{0}=\sqrt{\frac{L}{C}}$, which has
units of ohms ($\texttt{NsC}^{-2}$), we have annihilation operator 
\begin{eqnarray*}
a=\left( 2\hbar Z_{0}\right) ^{-1/2}\left( Z_{0}q+ip\right)
\end{eqnarray*}
so that $\left[ a,a^{\ast }\right] =1$. Here we have $ q=\sqrt{\frac{\hbar }{2Z_{0}}}\left( a+a^{\ast }\right)$ and
$ p=-i\sqrt{%
\frac{\hbar Z_{0}}{2}}\left( a-a^{\ast }\right) $, while the Hamiltonian is then a $H=\hbar \omega _{0}\left( a^{\ast }a+\frac{1}{2}%
\right) $ where the resonant frequency is $\omega _{0}=1/\sqrt{LC}$.
\end{remark}

The field is quantized by imposing equal time canonical commutation
relations 
\begin{eqnarray}
\left[ Q_j(z,t),\pi_k (z^{\prime },t)\right] =i\hbar\delta_{jk} \,\delta (z-z^{\prime }),
\label{eq:CCR_Q_pi}
\end{eqnarray}
however, this needs to be understood as valid for $z,z^\prime >0$ \emph{only}!

This should be equivalent to the commutation relations
\begin{eqnarray}
\left[ a_{\mathrm{in},j }(\omega ),a_{\mathrm{in},k}(\omega ^{\prime })^{\ast }\right]
=\delta _{jk }\,\delta (\omega -\omega ^{\prime }).
\label{eq:CCR_LR}
\end{eqnarray}
We will comment on this further in the following section.

The field Hamiltonian is then 
\begin{eqnarray*}
H_{\mathrm{field}}= 
\int_{0}^{\infty }\hbar \omega \, \mathbf{a}_{\mathrm{in}}(\omega )^{\ast }
\mathbf{a}_{\mathrm{in}}(\omega )d\omega .
\end{eqnarray*}

\section{Models With Memory}
\label{sec:Memory}

Memory effects are incorporated by replacing (\ref{eq:Langevin}) with the
Langevin equation, see for instance \cite{Weiss}, \cite{FLC88},
\begin{eqnarray}
\mathbf{L\ddot{q}}(t)+\int_{-\infty }^{t}\mathbf{\Gamma }\left( t-t^{\prime
}\right) \mathbf{\dot{q}}(t^{\prime })dt^{\prime }+\mathbf{Kq}(t)=\mathbf{F}%
(t),  \label{eq:Langevin_memory}
\end{eqnarray}
where now $\mathbf{\Gamma }(t)$ is the \textit{memory kernel}.

Denoting Laplace transforms as 
\[
\mathscr{L}f\left[ s\right] \triangleq \int_{0}^{\infty }f\left( t\right)
e^{-st}dt, 
\]
we have the relation $\mathscr{L}\mathbf{q}\left[ s\right] =%
\mathbf{\chi }_{-}\left[ s\right] \,\mathscr{L}\mathbf{F}\left[ s\right] $
where the $\textbf{F}$-to-$\textbf{q}$ transfer function matrix is
\begin{eqnarray*}
\mathbf{\chi }_{-}\left[ s\right] =\frac{1}{\mathbf{L}s^{2}+s\,\mathscr{L}%
\mathbf{\Gamma }\left[ s\right] +\mathbf{K}},\quad \text{Re}s>0.
\end{eqnarray*}
From this we obtain the susceptibility as the boundary value of $\mathbf{%
\chi }_{-}\left[ s\right] $:
\begin{eqnarray*}
\mathbf{\alpha }_{-}\left( \omega \right) =\mathbf{\chi }_{-}\left[
0^{+}-i\omega \right] =\frac{1}{-\mathbf{L}\omega ^{2}-i\omega \,\mathbf{R}%
\left( \omega \right) +\mathbf{K}}
\end{eqnarray*}
where we introduce the frequency dependent resistance
\begin{eqnarray*}
\mathbf{R}\left( \omega \right) \triangleq \mathscr{L}\mathbf{\Gamma }\left[
0^{+}-i\omega \right] .
\end{eqnarray*}
In \cite{FLC88} it is shown that for thermal fields, thermodynamic stability in ensured by the criterion that the Laplace transform of the memory kernel is lossless positive real. 

Similarly, the transfer function $\mathbf{G}\left( \omega \right) $ is the
boundary value of the $\mathbf{S}\left[ s\right] $ which is the analytic
function in the right hand plane defined by
\begin{eqnarray}
\mathbf{S}\left[ s\right] \triangleq \frac{\mathbf{L}s^{2}-s\,\mathscr{L}%
\mathbf{\Gamma }\left[ s\right] +\mathbf{K}}{\mathbf{L}s^{2}+s\,\mathscr{L}%
\mathbf{\Gamma }\left[ s\right] +\mathbf{K}},\quad \text{Re}s>0,
\label{eq:Scattering_Matrix}
\end{eqnarray}
that is 
\begin{eqnarray}
\mathbf{G}\left( \omega \right) =\mathbf{S}\left( 0^{+}-i\omega
\right) .
\label{eq:G_S}
\end{eqnarray} 

\subsection{Linear Passive Restriction}
In this section we wish to look at control-theoretic concepts of dissipativity \cite{Willems}. In particular, we reformulate the stability of the model as the lossless bounded real property of $S$.  We recall the basic definitions, \cite{Anderson_Vongpanitlerd}.

\textbf{Definition:} A function $\mathbf{\Sigma }[s]$ of a complex variable $%
s$ is said to be \textit{positive real} if it is analytic in the open right
hand plane ($s$ with Re $s>0$) with
\begin{eqnarray*}
\mathrm{Re\,}{\Sigma }\left[ s\right] \geq 0\quad \left( \text{Re } \,
s>0\right) ,
\end{eqnarray*}
and if its boundary function $\omega \mapsto \mathbf{\Sigma }\left(
0^{+}-i\omega \right) $ is well-defined for almost all real values $\omega
\in \mathbb{R}$ and satisfies 
\begin{eqnarray*}
\mathrm{Re\,}{\Sigma }\left( 0^{+}-i\omega \right)  &\geq &0, \\
\mathbf{\Sigma }\left( 0^{+}-i\omega \right) ^{\ast } &=&\mathbf{\Sigma }%
\left( 0^{+}+i\omega \right) ,
\end{eqnarray*}
for all $\omega \in \mathbb{R}$ with $i\omega $ not a pole of $\mathbf{%
\Sigma }$. If furthermore $\mathrm{Re\,}{\Sigma }\left( 0^{+}-i\omega
\right) =0$, for $i\omega $ not a pole of $\mathbf{\Sigma }$, then we say it
is \textit{lossless positive real} (LPR).

\bigskip 

\textbf{Definition:} A function $\mathbf{S}[s]$ of a complex variable $s$ is
said to be \textit{bounded real} if it is analytic in the open right hand
plane ($s$ with Re $s>0$) with
\begin{eqnarray*}
\mathbf{S}[s]^{\ast }\mathbf{S}[s]\leq I\quad \left( \text{Re\thinspace }%
s>0\right) ,
\end{eqnarray*}
and if its boundary function $\omega \mapsto \mathbf{S}\left( 0^{+}-i\omega
\right) $ is well-defined for almost all real values $\omega \in \mathbb{R}$%
. If furthermore $\mathrm{Re\,}{\Sigma }\left( 0^{+}-i\omega \right) =0$,
for $i\omega $ not a pole of $\mathbf{\Sigma }$, then we say it is \textit{%
lossless bounded real} (LBR).

\begin{proposition}
Let us suppose that we have the relation
\begin{eqnarray}
\mathbf{S}[s]=\frac{I-\mathbf{\Sigma }[s]}{I+\mathbf{\Sigma }[s]}
\label{eq:S_Sigma}
\end{eqnarray}
then $\mathbf{S}[s]$ is LBR if and only if $\mathbf{\Sigma }[s]$ is LPR.
\end{proposition}
This is a well known result in linear circuit theory. A proof can be found in \cite{Anderson_Vongpanitlerd}, Theorem 2.7.4
where it is established that an immittance matrix is LPR if and only if the scattering matrix
is LBR, thereby relating the lossless property of an immittance matrix to that of
the scattering matrix. See also \cite{ZGPG} for a more complete proof.

\bigskip

\begin{theorem}
\label{thm:S_LBR}
The scattering matrix $\mathbf{S}$ given by (\ref{eq:Scattering_Matrix}) is
LBR if and only if the Laplace transform of the memory kernel is LPR.
\end{theorem}

\begin{proof}
This is a straightforward corollary since (\ref{eq:S_Sigma}) can be
rearranged as $\mathbf{\Sigma }[s]=\frac{I-\mathbf{S}[s]}{I+\mathbf{S}[s]}$,
and substituting in for (\ref{eq:Scattering_Matrix}) we have specifically
\begin{eqnarray*}
\mathbf{\Sigma }[s]=\frac{2s }{\mathbf{L}s^{2}+\mathbf{K}} \,\mathscr{L}\mathbf{\Gamma }\left[ s\right]
\end{eqnarray*}
so evidently $\mathbf{\Sigma }$ is LPR if and only if $\mathscr{L}\mathbf{%
\Gamma }$ is LPR.
\end{proof}

\bigskip 

\subsection{Spectral Density of the Network}

Let us make the postulate that the scattering matrix is LBR, then we define
the \textit{spectral density matrix} of the network to be
\begin{eqnarray*}
\mathbf{J}\left( \omega \right) \triangleq \omega \,\text{\textrm{Re }}%
\mathbf{R}\left( \omega \right) .
\end{eqnarray*}
By assumption it follows that $\mathbf{J}\left( \omega \right) \geq 0$ where
defined. The spectral density matrix actually determines $\mathbf{R}\left(
\omega \right) $ as, from the Kramers-Kronig relations, we have 
\begin{eqnarray*}
\omega \,\mathrm{Im}\,\mathbf{R}\left( \omega \right) =-\frac{1}{\pi }%
\mathtt{PV}\int \frac{\mathbf{J}(\omega ^{\prime })d\omega ^{\prime }}{%
\omega -\omega ^{\prime }}
\end{eqnarray*}
and combining the two gives 
\begin{eqnarray*}
-i\omega \mathbf{R}\left( \omega \right) =-i\mathbf{J}(\omega )-\frac{1}{\pi 
}\mathtt{PV}\int \frac{\mathbf{J}(\omega ^{\prime })d\omega ^{\prime }}{%
\omega -\omega ^{\prime }}.
\end{eqnarray*}
The memory kernel is then 
\begin{eqnarray*}
\mathbf{\Gamma }\left( t\right) =\frac{1}{\pi }\int_{0}^{\infty }\frac{%
\mathbf{J}\left( \omega \right) }{\omega }\cos \omega t\,d\omega \;\theta
\left( t\right) .
\end{eqnarray*}

\begin{proposition}
In the case that the memory function corresponds to several independent input channels, that is, it takes the diagonal form
$\mathbf{\Gamma }\left( t\right) \equiv \mathrm{diag}(\Gamma
_{1}\left( t\right) ,\cdots ,\Gamma _{n}\left( t\right) )$, then the spectral density will be likewise 
diagonal:
\begin{eqnarray}
\mathbf{J}\left( \omega \right) \equiv \mathrm{diag}(J_{1}\left( \omega
\right) ,\cdots ,J_{n}\left( \omega \right) ) .
\label{eq:J_diag}
\end{eqnarray}
\end{proposition}

The function $J_k$ is then called the \emph{spectral density} for the $k$th transmission line.

\subsection{Hamiltonian Models}
For simplicity we assume that the spectral density is diagonal, as in (\ref{eq:J_diag}).
We consider the Hamiltonian 
\begin{eqnarray}
H &=&\tilde{H}_{\mathrm{comp.}} \nonumber \\
&+&\int d\omega \sum_k
\Bigg\{ 
\frac{\omega }{2J_k\left( \omega \right) }  \left[ \hat{\pi}_k\left( \omega
\right) -\sqrt{\frac{2}{\pi }}\frac{J_k\left( \omega \right) }{\omega }q\right]
^{2}   \nonumber\\
&+&\frac{1}{2} \omega J_k\left( \omega \right) \hat{Q}_k \left( \omega
\right) ^{2} \Bigg\}   \label{eq:Caldeira_Leggett}
\end{eqnarray}
where $\tilde{H}_{\mathrm{comp.}}=\frac{1}{2}\mathbf{p}^\top \mathbf{L}^{-1}\mathbf{p} + \frac{1}{2} \mathbf{q}^\top \tilde{\mathbf{K}} \mathbf{q}$, and we 
introduce canonically conjugate field operators satisfying
\begin{eqnarray*}
\left[ \hat{Q}_j\left( \omega \right) ,\hat{\pi}_k\left( \omega ^{\prime
}\right) \right] =i\hbar \delta_{jk} \,  \delta \left( \omega -\omega ^{\prime }\right) .
\end{eqnarray*}

The Hamiltonian takes the form 
\begin{eqnarray*}
H=H_{\mathrm{comp.}}+H_{B}+H_{\mathrm{int}}
\end{eqnarray*}
with 
\[ H_{\mathrm{comp.}}=\tilde{H}_{S}+\frac{1}{2}\sum_k K_k q_k^{2}\equiv \frac{1}{2}\mathbf{p}^\top \mathbf{L}^{-1}\mathbf{p} + \frac{1}{2} \mathbf{q}^\top \mathbf{K} \mathbf{q}, 
\]
where
$K_k=\frac{2}{\pi }\int \frac{J_k \left( \omega \right) }{\omega }d\omega $ 
(assumed finite for each $k$) so that $\mathbf{K}= \tilde{ \mathbf{K}} +\mathrm{diag} (K_1, \cdots , K_n )$,
and 
\begin{eqnarray*}
H_{B} &=& \hspace{-0.1cm}\sum_k \int \left\{ \frac{\omega }{%
2J_k \left( \omega \right) }\hat{\pi}_k\left( \omega \right) ^{2}+\frac{\omega J_k\left( \omega \right)}{2}
 \hat{Q}_k \left( \omega \right) ^{2}\right\}
d\omega , \\
H_{\mathrm{int}} &=& \hspace{-0.1cm}-\sqrt{\frac{2}{\pi }} \sum_k q_k \int\hat{\pi}_k \left( \omega \right) d\omega .
\end{eqnarray*}

We note that we may introduce the annihilator density 
\begin{eqnarray*}
a_k\left( \omega \right) = \sqrt{ \frac{J_k\left( \omega \right) }{2\hbar }} \hat{%
Q}_k\left( \omega \right) +i\frac{1}{\sqrt{2\hbar J_k\left( \omega \right) }}%
\hat{\pi}_k\left( \omega \right)
\end{eqnarray*}
so that 
\begin{eqnarray*}
\hat{Q}_k\left( \omega \right) &=&\sqrt{\frac{\hbar }{2J_k\left( \omega \right) 
}}\big\{a_k\left( \omega \right) +a_k\left( \omega \right) ^{\ast } \big\}, \\
\hat \pi_k \left( \omega \right) &=&-i\sqrt{\frac{\hbar J_k\left( \omega
\right) }{2}}\big\{ a_k\left( \omega \right) -a_k\left( \omega \right) ^{\ast } \big\}
\end{eqnarray*}
where $\left[ a_j\left( \omega \right) ,a_k\left( \omega ^{\prime }\right)
^{\ast }\right] =\delta_{jk} \delta \left( \omega -\omega ^{\prime }\right) $, and in
these terms we have 
\begin{eqnarray*}
H_{B} &=&\sum_k \int_{0}^{\infty }\hbar \omega \, a_k\left( \omega \right) ^{\ast
}a_k\left( \omega \right) d\omega , \\
H_{\mathrm{int}} &=&i \sum_k q_k \int_{0}^{\infty }\sqrt{\frac{\hbar J_k\left( \omega
\right) }{\pi }} \big\{a_k\left( \omega \right) -a_k \left( \omega \right) ^{\ast
} \big\}d\omega .
\end{eqnarray*}

The equations of motion are 
\begin{eqnarray*}
\dot{\mathbf{q}}_{t} &=& \mathbf{L}^{-1} \mathbf{p}_{t}, \\
\dot{\mathbf{p}}_{t} &=& -\mathbf{K} \mathbf{ q}_{t} +\sqrt{\frac{2}{\pi }%
}\int \hat{\mathbf{\pi}}_{t}\left( \omega \right)
d\omega , \\
\frac{d}{dt}\hat{ Q}_{k,t}\left( \omega \right) &=&\frac{\omega }{J_k\left(
\omega \right) }\hat{\pi}_{k,t}\left( \omega \right) -\sqrt{\frac{2}{\pi }}
q_{k,t}, \\
\frac{d}{dt}\hat{\pi}_{k,t}\left( \omega \right) &=&-\omega J_k\left( \omega
\right) \hat{Q}_{k,t}\left( \omega \right) .
\end{eqnarray*}
(The last pair of equations are equivalent to $\dot{a}_{k,t}\left( \omega
\right) =-i\omega a_{k, t}\left( \omega \right) -\sqrt{\frac{J_k\left( \omega
\right) }{\pi \hbar }} q_{k,t}$.) These lead to 
\begin{eqnarray}
\mathbf{L} \ddot{\mathbf{q}}_{t}+ \mathbf{K} \mathbf{q}_{t} =\sqrt{\frac{2}{\pi }}
\int \hat{\mathbf{\pi} }_{t}\left( \omega \right)
d\omega ,  \label{eq:Ham_qp}
\end{eqnarray}
and 
\begin{eqnarray}
\frac{d^{2}}{dt^{2}}\hat{\pi}_{k,t}\left( \omega \right) +\omega ^{2}\hat{\pi}%
_{k,t}\left( \omega \right) =\sqrt{\frac{2}{\pi }}\omega J_k \left( \omega
\right) q_{k,t}.  \label{eq:Ham_pi}
\end{eqnarray}
The equation (\ref{eq:Ham_pi}) is a linear homogeneous equation in $\hat{\pi}%
_{t}\left( \omega \right) $ and so has solution of the form general
homogeneous solution plus particular solution: 
\begin{eqnarray*}
&&\hat{\pi}_{k,t}\left( \omega \right) = \hat{\pi}_{k,t}^{h}\left( \omega \right) \nonumber\\
&+ &
\sqrt{\frac{2}{\pi }}\frac{J_k \left( \omega \right) }{\omega }\left[
q_{k,t}-\int_{-\infty }^{t}\cos \omega \left( t-t^{\prime }\right) \, \dot{q}%
_{k, t^{\prime }}dt^{\prime }\right]
\end{eqnarray*}
where 
\begin{eqnarray*}
\hat{\pi}_{k,t}^{h}\left( \omega \right) =-i\sqrt{\frac{\hbar J_k\left( \omega
\right) }{2}}\big\{ a_k\left( \omega \right) e^{-i\omega t}-a_k\left( \omega \right)
^{\ast }e^{i\omega t} \big\}
\end{eqnarray*}
Substituting this representation into (\ref{eq:Ham_qp}) yields the \textit{generalized quantum Langevin
equation} 
\begin{eqnarray}
\mathbf{L} \ddot{\mathbf{q}}_{t}+\int \mathbf{\Gamma } \left( t-t^{\prime }\right) \dot{\mathbf{q}}_{t^{\prime
}}dt^{\prime }+ \mathbf{K} \mathbf{q}_{t}
=\mathbf{ F}\left( t\right)
\label{eq:genQLE}
\end{eqnarray}
with the \textit{memory kernel} $\mathbf{\Gamma} (t)= \frac{2}{\pi }\int \frac{\mathbf{J}\left( \omega \right) }{\omega }\cos \omega t\,d\omega
\;\theta \left( t\right) $ and the \textit{generalized
Langevin force} $\mathbf{F}(t) \equiv \sqrt{\frac{2}{\pi}} \int \hat{\mathbf{\pi}}_t (\omega ) d\omega$, that is
\begin{eqnarray}
\Gamma_{jk} \left( t\right) &=& \delta_{jk} \, \frac{2}{\pi }\int \frac{J_k\left( \omega \right) }{\omega }\cos \omega t\,d\omega
\;\theta \left( t\right) ,  \label{eq:memory} \\
F_k \left( t\right) &=&-i\int\sqrt{\frac{%
\hbar J_k \left( \omega \right) }{\pi }}\left( a_k \left( \omega \right)
e^{-i\omega t}-\mathrm{H.c.}\right) d\omega .  \nonumber \\
&&  \label{eq:Langevin_force}
\end{eqnarray}
Here $\theta $ is the Heaviside function, so that the memory function is causal, that is $\mathbf{\Gamma}
\left( t\right) =0$ for $t<0$.

\subsection{Input-Output Relations with Memory}

As before we may write the charge at position $z>0$ and at time $t$ on the
transmission wire as 
\begin{eqnarray*}
\mathbf{Q} \left( t,z\right) =\mathbf{q}_{\mathrm{in}}\left( t+\frac{z}{c}\right) +
\mathbf{q}_{\mathrm{out}} \left( t-\frac{z}{c}\right)
\end{eqnarray*}
where now, generalizing (\ref{eq:free a}),
\begin{eqnarray*}
&& q_{	\mathrm{in/out},k }\left( \tau \right) =\\
&&\int 
\sqrt{\frac{\hbar }{4\pi J_k \left( \omega \right) }}\left[ a_{\mathrm{in/out},k  }(\omega
)e^{-i\omega \tau }+\mathrm{H.c.}\right] d\omega
\end{eqnarray*}
with $a_{\mathrm{in},k}\left( \omega \right) $
identified with $a_k\left( \omega \right) $ above. As before we have 
\begin{eqnarray*}
\mathbf{q}_{\mathrm{out}}(t )= \mathbf{q}(t)-\mathbf{q}_{\mathrm{in}}(t )
\end{eqnarray*}
and we may again determine $\mathbf{q} \left( t\right) $ as a Fourier transform, this
time by solving the generalized Langevin equation (\ref{eq:genQLE}) to get $%
\hat{\mathbf{q}}\left( \omega \right) =\mathbf{\alpha}_- \left( \omega \right) \, \hat{\mathbf{F}}\left(
\omega \right) $ with a generalized susceptibility 
\begin{eqnarray}
\mathbf{\alpha}_- \left( \omega \right) =\left[ \mathbf{K}-\mathbf{L}\omega ^{2}-i\omega  \mathbf{R }\left( \omega \right) \right] ^{-1}  \label{eq:chi}
\end{eqnarray}
and with $\hat{F}_k\left( \omega \right) =-i\sqrt{\frac{\hbar J_k\left( \omega
\right) }{\pi }}\left( a_k \left( \omega \right) e^{-i\omega t}-\mathrm{H.c.}
\right) $. We now obtain 
\begin{eqnarray}
\mathbf{a}_{\mathrm{out}}\left( \omega \right) =\mathbf{G} \left( \omega \right) \mathbf{a}_{\mathrm{in}%
}\left( \omega \right)  \label{eq:io_G}
\end{eqnarray}
where, $\mathbf{G} \left( \omega \right)= 2i\mathbf{\alpha}_-\left( \omega \right) \mathbf{J}\left( \omega \right) -I_n$,
which is given by 
\begin{eqnarray}
\mathbf{G}\left( \omega \right) =-
\frac{\mathbf{K} - \mathbf{L} \omega ^{2}+i\omega  \mathbf{R }\left( \omega \right)^\ast }
{\mathbf{K} - \mathbf{L} \omega ^{2}-i\omega  \mathbf{R }\left( \omega \right)}
,
\end{eqnarray}
in agreement with (\ref{eq:G_S}).

The transfer function is unimodular, so that $a_{\mathrm{out}}(\omega )$
again satisfies the canonical commutation relations: 
\begin{eqnarray*}
\left[ a_{\mathrm{out}}(\omega ),a_{\mathrm{out}}(\omega ^{\prime })^{\ast }%
\right] =\delta \left( \omega -\omega ^{\prime }\right) .
\end{eqnarray*}

\subsubsection{Commutation Relations}
We find the commutation relations
\begin{eqnarray*}
&&[Q_j (z_1 , t_1 ) , Q_k (z_2 ,t_2) ] \\
&=& \frac{i \hbar}{2 \pi} \int d \omega \frac{d \omega}{ \sqrt{ J_j(\omega ) J_k (\omega )}} \\
&& \times \mathrm {Im}  \bigg\{ e^{-i \omega (t_1 -t_2 )} \bigg( 2\delta_{jk}  \cos \frac{ \omega (z_1 - z_2)}{c} \\
&&+G_{jk} (\omega ) e^{i \omega (z_1 +z_2 ) /c} +G_{kj}^\ast (\omega ) e^{-i \omega (z_1 +z_2 ) /c}
\bigg)  \bigg\}.
\end{eqnarray*}
In particular we have the equal time commutation relations
\begin{eqnarray*}
[Q_j (z_1 , t ) , Q_k (z_2 ,t) ]=0,
\end{eqnarray*}
and the terminal charge commutation relations
\begin{eqnarray*}
&&[Q_j (0 , t_1 ) , Q_k (0 ,t_2) ] \\
&=& \frac{i \hbar}{2 \pi} \int d \omega \frac{d \omega}{ \sqrt{ J_j(\omega ) J_k (\omega )}} \\
&& \times \mathrm {Im}  \bigg\{ e^{-i \omega (t_1 -t_2 )} \bigg( 2\delta_{jk}  
+G_{jk} (\omega )  +G_{kj}^\ast (\omega )
\bigg)  \bigg\}.
\end{eqnarray*}

By direct substitution we see that the input fields satisfy the CCR
\begin{eqnarray}
\frac{1}{i \hbar} [ q_{ \mathrm{in},j} (t) , q_{\mathrm{in}, k} (t^\prime ) ] = \delta_{jk} \, \sigma_k (t - t^\prime ),
\end{eqnarray}
where
\begin{eqnarray}
\sigma_k (\tau) = - \frac{1}{2\pi } \int \frac{ \sin  (\omega \tau )}{ J_k (\omega )}  d \omega.
\end{eqnarray}
From the relations (\ref{eq:io_G}) and the unitarity of $\mathbf{G}$, we likewise have
\begin{eqnarray}
\frac{1}{i \hbar} [ q_{ \mathrm{out},j} (t) , q_{\mathrm{out}, k} (t^\prime ) ] = \delta_{jk} \, \sigma_k (t - t^\prime ),
\end{eqnarray}

\begin{remark}
In the $n=1$ ohmic case we find the specific form $\sigma _{\mathrm{ohmic}} (\tau ) = - \frac{1}{4R} \mathrm{sign} (\tau )$.
\end{remark}

\subsubsection{Input-Output Causality}
We also find 
\begin{eqnarray*}
\frac{1}{i \hbar} \left[ q_{\mathrm{in},j}(t),q_{\mathrm{out},k }(t^{\prime })\right]
=g_{jk}(t-t^{\prime })
\end{eqnarray*}
where 
\begin{eqnarray*}
g_{jk} (\tau ) &=&- \frac{1}{2\pi }\int \frac{1}{ \sqrt{J_j (\omega ) J_k (\omega )}} 
\mathrm{Im} \bigg\{ e^{-i\omega \tau } G_{kj} (\omega) \bigg\}  d\omega .
\end{eqnarray*}
Recalling that $\mathbf{G} (\omega ) $ is the boundary function $s=0^+ -i \omega$ of the 
scattering matrix $S [s]$ which we require to be LBR, and in particular analytic in the right hand plane.
It follows that the functions $g_{jk}$ are \textit{causal}: the commutator vanishes if $%
t>t^{\prime }$. This means that past output $q_{\mathrm{out}}(t^{\prime })$
will commute with future input $q_{\mathrm{in}}(t)$.

\subsubsection{The Langevin Force}
The commutation relations for the force are
\begin{eqnarray*}
\frac{1}{i\hbar } \left[ F_j(t),F_k (t^{\prime })\right] =-\delta_{jk} \frac{2}{\pi }
\int J_k(\omega )\sin \omega \left(
t-t^{\prime }\right) \,d\omega .
\end{eqnarray*}

We may take a thermal state for the field, so that 
\begin{eqnarray*}
\langle a_{\mathrm{in},j}\left( \omega \right) ^{\ast }a_{\mathrm{in},k}\left(
\omega ^{\prime }\right) \rangle = \delta_{jk} n_k\left( \omega \right) \delta \left(
\omega -\omega ^{\prime }\right)
\end{eqnarray*}
with $n_j\left( \omega \right) =\frac{1}{e^{\beta_j \hbar \omega }-1}$. The
result is that 
\begin{eqnarray*}
\langle F_j\left( t\right) F_k\left( t^{\prime }\right) \rangle &=&\frac{\hbar 
}{\pi }\int_{0}^{\infty } J_k (\omega )\bigg\{ e^{-i\omega \left( t-t^{\prime }\right)
}\left( n_k\left( \omega \right) +1\right) \\
&& \qquad +e^{i\omega \left( t-t^{\prime }\right) }n_k\left( \omega \right) \bigg\}d\omega .
\end{eqnarray*}
From this we see that the symmetrized two-point correlation is 
\begin{eqnarray*}
&&\frac{1}{2}\langle F_j\left( t\right) F_k\left( t^{\prime }\right) +F_k\left(
t^{\prime }\right) F_j\left( t\right) \rangle \\
&=& \frac{\hbar }{\pi } \delta_{jk} \int J_k(\omega )
\coth \frac{\beta_k \hbar \omega }{2}\,\cos \omega \left( t-t^{\prime
}\right) \,d\omega .
\end{eqnarray*}
This is effectively the Ford-Kac-Mazur model \cite{FKM65} for quantum
dissipation (without memory effects), see also \cite{FLC88} for the non-ohmic case.

\section{Markov Limit}
\label{sec_Markov}

In this section we consider a Markov limit of open systems treated up to
now. Derivations of the Markov limit as a broadband approximation have been given
in \cite{GarZol00} and \cite{Yurke_squeezing}, however it is advantageous to provide a
mathematically rigorous account giving the requisite scaling regime. Our approach will be to to make a weak coupling approximation,
specifically we make the replacements 
\begin{eqnarray*}
J_{k}\left( \omega \right) &\hookrightarrow &\lambda ^{2}\,J_{k}\left(
\omega \right) , \\
t &\hookrightarrow &t/\lambda ^{2}
\end{eqnarray*}
with $\lambda \rightarrow 0$. This is the van Hove limit corresponding to an
interaction of strength $\lambda $ which has vanishing average effect as a
first order perturbation but builds up a non-trivial contribution second
order effect when looked at over long times scales (order $\lambda ^{-2}$).
Once this pre-limit has been identified, we have available the limit theorems of
\cite{AFL}, \cite{AGL} which capture the Markovian limit. For simplicity we shall
work in the vacuum state for the field, but the limit for general Gaussian states
such as thermal or squeezed states is also known in this context.

We shall assume that the component Hamiltonian may be written in mode form
as 
\begin{eqnarray*}
H_{\text{comp.}}=\tilde{H}_{S}+\frac{\lambda ^{2}}{2}\sum_{k}K_{k}q_{k}^{2}=%
\sum_{k}\hbar \Omega _{k}a_{k}^{\ast }a_{k}+O\left( \lambda ^{2}\right)
\end{eqnarray*}
where we will now ignore the correction terms which are order $\lambda ^{2}$.

We also assume that the terminal charges may be written as 
\begin{eqnarray*}
q_{k}=\sqrt{\frac{\hbar }{2}}\left( \sum_{j}Y_{kj}a_{j}+\text{H.c.}\right) 
\end{eqnarray*}
where the $Y_{kj}$ are suitable constants. The interaction takes the form 
\begin{eqnarray*}
H_{\text{int}}&=&-i\hbar \sum_{jk}\left( Y_{kj}a_{j}+\text{H.c.}\right)\\
&&
\otimes \int d\omega \sqrt{\frac{J_{k}\left( \omega \right) }{2\pi }}\left\{
a_{k}\left( \omega \right) -a_{k}\left( \omega \right) ^{\ast }\right\} 
\end{eqnarray*}

The next step is to move to the interaction picture: we have the
Hamiltonians 
\begin{eqnarray*}
H_{0} &=&\tilde{H}_{S}+H_{B}= \\
&&\sum_{k}\hbar \Omega _{k}a_{k}^{\ast }a_{k}+\sum_{k}\int \hbar \omega
a_{k}\left( \omega \right) ^{\ast }a_{k}\left( \omega \right) d\omega , \\
H_{\lambda } &=&H_{0}+\lambda H_{\text{int}},
\end{eqnarray*}
and the unitary transforming to the interaction picture is $e^{itH_{0}/\hbar
}e^{-itH_{\lambda }/\hbar }$ and with rescaled time we have 
\begin{eqnarray*}
U_{\lambda }\left( t\right) =e^{itH_{0}/\lambda ^{2}\hbar }e^{-itH_{\lambda
}/\lambda ^{2}\hbar }.
\end{eqnarray*}
We then have 
\begin{eqnarray*}
\dot{U}_{\lambda }\left( t\right) =-i\Upsilon _{\lambda }\left( t\right)
U_{\lambda }\left( t\right) 
\end{eqnarray*}
where 
\begin{eqnarray}
&&-i\Upsilon _{\lambda }\left( t\right)  =-\frac{1}{\lambda }\sum_{jk}\left(
Y_{kj}a_{j}e^{-i\Omega _{k}t/\lambda ^{2}}+\text{H.c.}\right)   \nonumber  \\
&&\times \int d\omega \sqrt{\frac{J_{k}\left( \omega \right) }{2\pi }}%
\left\{ a_{k}\left( \omega \right) e^{-i\omega t/\lambda ^{2}}-a_{k}\left(
\omega \right) ^{\ast }e^{i\omega t/\lambda ^{2}}\right\} .
\nonumber\\
\label{eq:Upsilon}
\end{eqnarray}

\subsection{Quantum Markov Communication Channels}
Let us introduce the processes
\begin{eqnarray*}
B_{k,\Omega }\left( t,\lambda \right) &=& \frac{1}{\lambda }\int_{0}^{t}d\tau
\int d\omega \\
&& \sqrt{\frac{J_{k}\left( \omega \right) }{2\pi }}a_{k}\left(
\omega \right) e^{-i(\omega -\Omega )\tau /\lambda ^{2}}
\end{eqnarray*}
then we note that
\begin{eqnarray*}
&&\left[ B_{k,\Omega }\left( t,\lambda \right) ,B_{k^{\prime },\Omega
^{\prime }}\left( t^{\prime },\lambda \right) \right]  \\
&=&\delta _{k,k^{\prime }}\frac{1}{\lambda ^{2}}\int_{0}^{t}d\tau
\int_{0}^{t^{\prime }}d\tau ^{\prime }e^{i\left( \Omega -\Omega ^{\prime
}\right) \tau ^{\prime }/ \lambda ^{2}}\\
&&\int d\omega \frac{J_{k}\left( \omega
\right) }{2\pi }e^{-i\left( \omega -\Omega \right) \left( \tau -\tau
^{\prime }\right) /\lambda ^{2}}
\end{eqnarray*}
which converges in the limit $\lambda \rightarrow 0$ to
\begin{eqnarray*}
\delta _{k,k^{\prime }}\delta _{\Omega ,\Omega ^{\prime }}J_{k}\left( \Omega
\right) \,\min \left\{ t,t^{\prime }\right\} .
\end{eqnarray*}
Here we note that $\frac{1}{\lambda ^{2}}e^{-i\left( \omega -\Omega \right)
\left( \tau -\tau ^{\prime }\right) /\lambda ^{2}}\rightarrow 2\pi \,\delta
\left( \omega -\Omega \right) $.

The $B_{k,\Omega }\left( t,\lambda \right) $ are converging to quantum
Wiener processes. We see that different frequencies lead to independent
process on account of the rapidly oscillating phase $e^{i\left( \Omega
-\Omega ^{\prime }\right) \tau ^{\prime }/\lambda ^{2}}$: assuming that we
have a finite number of frequencies $\Omega $, then this follows from the
Riemann-Lebesgue lemma.

We also see that in (\ref{eq:Upsilon}) we may similarly ignore the
contributions from terms $e^{\pm i\left( \Omega _{k}+\omega \right)
t/\lambda ^{2}}$ as these have an associated $2\pi \delta \left( \omega
+\Omega _{k}\right) $ distribution, but $\Omega _{k}>0$ and the spectral
densities are assume to vanish for negative frequencies. Dropping such terms
constitutes the rotating wave approximation and we may rigorously show that
the same limit is obtained when $-i\Upsilon _{\lambda }\left( t\right)$ in 
(\ref{eq:Upsilon}) is replaced by
\begin{eqnarray*}
\frac{1}{\lambda }\sum_{jk}
Y_{kj}a_{j}\int d\omega \sqrt{\frac{J_{k}\left( \omega \right) }{2\pi }}%
a_{k}\left( \omega \right) ^{\ast }e^{i(\omega -\Omega _{k})t/\lambda
^{2}} -\text{H.c.}
\end{eqnarray*}

What one may show rigorously is that the limit evolution (in the sense of
weak convergence of matrix elements) defined on the joint Hilbert space of
the system plus a noise space determined by quantum Wiener processes $%
B_{k,\Omega }\left( t\right) $ labelled by the transmission line indices $k$
and the harmonic frequencies $\Omega $ of the component. We have the quantum It\={o} table
\begin{eqnarray}
dB_{k, \Omega} (t) \, dB_{k^\prime , \Omega^\prime } (t) = \delta_{k , k^\prime } \delta_{\Omega , \Omega^\prime} \,
J_k (\Omega ) \, dt ,
\end{eqnarray}
with all other products of increments vanishing. Note that the total number of independent quantum Wiener processes
equals the number of transmission lines times the number of distinct frequencies $\Omega_k$ of the component.

The limit quantum stochastic differential equation is
\begin{eqnarray}
dU\left( t\right) & =& \bigg\{ \sum_{jk}Y_{kj}a_{j}\otimes dB_{k,\Omega _{j}}\left(
t\right) ^{\ast } \nonumber \\
&& -\sum_{jk}Y_{kj}^{\ast }a_{j}^{\ast }\otimes dB_{k,\Omega
_{j}}\left( t\right) +K\otimes dt \bigg\} U\left( t\right) \nonumber \\
\label{eq:QSDE}
\end{eqnarray}
where
\begin{eqnarray}
K &=& -\sum_{k}\sum_{j,j^{\prime }}\delta _{\Omega _{j},\Omega _{j^{\prime
}}}Y_{kj}Y_{kj^{\prime }}^{\ast }a_{j}^{\ast }a_{j^{\prime }} \nonumber \\
&& \times \int \frac{%
J_{k}\left( \omega \right) }{2\pi }\frac{d\omega }{i\left( \omega -\Omega
_{j}-i0^{+}\right) }.
\end{eqnarray}
In particular, we note that 
\begin{eqnarray*}
K+K^{\ast }=-\sum_{k}\sum_{j,j^{\prime }}\delta _{\Omega _{j},\Omega
_{j^{\prime }}}J_{k}\left( \Omega _{j}\right) Y_{kj}Y_{kj^{\prime }}^{\ast
}a_{j}^{\ast }a_{j^{\prime }}.
\end{eqnarray*}

\subsection{Single Input Case}
For clarity we consider the case of a single transmission line. Here the
terminal charge observable may be written as 
\begin{eqnarray*}
q=\sqrt{\frac{\hbar C}{2L}}\left( a+a^{\ast }\right) 
\end{eqnarray*}
so that $H_{\text{comp.}}=\hbar \Omega a^{\ast }a$ with $\Omega =\left(
LC\right) ^{-1/2}$ and
\begin{eqnarray*}
H_{\text{int}}=-i\lambda \hbar a\otimes \int \sqrt{\frac{J\left( \omega
\right) C}{2\pi L}}a\left( \omega \right) ^{\ast }+\text{H.c.}
\end{eqnarray*}
The associated QSDE for the limit is
\begin{eqnarray*}
dU\left( t\right) &=& \big\{ \sqrt{\gamma }a\otimes dB\left( t\right) ^{\ast }-%
\sqrt{\gamma }a^{\ast }\otimes dB\left( t\right)\\
&& -\kappa a^{\ast }a\otimes
dt\big\} U\left( t\right) 
\end{eqnarray*}
where
\begin{eqnarray*}
\kappa  &=&\int \frac{J\left( \omega \right) C}{2\pi L}\left\{ \pi \delta
\left( \omega -\Omega \right) +i\mathtt{PV}\frac{1}{\omega -\Omega }\right\} 
\\
&=&\frac{1}{2}\gamma +i\varepsilon .
\end{eqnarray*}
In particular
\begin{eqnarray*}
\gamma =2\text{Re}\left\{ \kappa \right\} =\frac{J\left( \Omega \right) C}{L}
\end{eqnarray*}
or $\gamma =R/L$ where $R=J\left( \Omega \right) /\Omega $. The imaginary component $\varepsilon$ of the damping corresponds to a shift of the resonant frequency, some times referred to as the Lamb shift \cite{Louisell}, \cite{AGL}, but is often negligble in size.

\subsection{General Input-State-Output Case}
The Heisenberg equations of motion for the mode $a_j$ are readily deduced from (\ref{eq:QSDE}). The observable describing the mode at time $t$ is
\begin{eqnarray}
a_j (t) \triangleq U(t)^\ast \left( a_j \otimes I \right) U(t),
\end{eqnarray}
and from the quantum It\={o} rule and (\ref{eq:QSDE}) we find
\begin{eqnarray}
d a_j (t)&=& - \sum_k \sum_{j^\prime} \delta_{\Omega_j , \Omega_{j^\prime}} \kappa_{kj} Y^\ast_{kj} Y_{kj^\prime} a_{j^\prime} (t)\, dt \nonumber \\
&& - \sum_k Y^\ast_{kj} dB_{k , \Omega_j} (t),
\label{eq:Heisenberg}
\end{eqnarray}
where $\kappa_{kj}= \int \frac{J_k ( \omega ) d \omega }{ 2 \pi i \, (\omega -\Omega_j -i0^+ )}$.

Similarly, the output fields are defined by 
\begin{eqnarray}
B^{\mathrm{out}}_{k, \Omega} (t) \triangleq U(t)^\ast \left(I\otimes B_{k ,\Omega} (t) \right) U(t),
\end{eqnarray}
and from the quantum It\={o} rule and (\ref{eq:QSDE}) we find
\begin{eqnarray}
d B^{\mathrm{out}}_{k, \Omega} (t) &=& d B_{k , \Omega} (t)
+ \sum_j   \delta_{\Omega_j , \Omega }     Y_{kj } a_{j } (t)\, dt ,
\label{eq:output}
\end{eqnarray}
where $\Omega $ is one of the resonant frequencies for the component.

It is worth emphasizing that (\ref{eq:Heisenberg}) and (\ref{eq:output}) are linear equations in the \lq\lq state variables"
$ a_j (t)$ as well as the Markovian input channels $B_{k, \Omega} (t)$ as well as the output channels $B^{\mathrm{out}}_{k, \Omega} (t)$. In fact they take the vectorized form
\begin{eqnarray}
d \mathbf{a} (t) &=& A \mathbf{a} (t) dt + B d \mathbf{B} (t)  \nonumber \\
d \mathbf{B}^{\mathrm{out}} (t) &=& C  \mathbf{a} (t) dt + d \mathbf{B} (t)
\end{eqnarray}
which takes the form of an $A-B-C-D$ linear input-state-output system. While these are fundamentally quantum systems, the obvious link to classical control systems has lead to a very natural theory of quantum control  \cite{YK}, \cite{GGY}-\cite{GZ} that has had significant impact on quantum engineering.

\section{Conclusions}
Our interest in non-Markovian models arises out of necessity - the natural models of classical/quantum transmission lines are not Markovian, even in the ohmic case where they lead to a $\delta$-correlated memory kernel. In particular, the standard theory of quantum filtering cannot apply. In \cite{XJSUP_15}, for instance, a tractable approach to filtering a non-Markov model is given but this relies on the being able to approximate the system as a subsystem of a Markov model. Ultimately, the extraction of information using quantum filtering techniques can only proceed if we have a Markov model. In our case it is clear that there is an idealized Markov model which approximates the non-Markov model - specifically we have Markov channels associated with each transmission line and labelled by the resonant frequencies $\Omega_j$ of the circuit. The field quanta are therefore those in transmission line that are approximately on the mass-shell ($\omega = \Omega_j$). A more detailed model would have to assume the existence of an appropriate oscillator in any measuring apparatus so as to set up an approximate Markov channel which is measured.

While transmission lines can give a wholly Hamiltonian model of dissipation it is worth recalling that they were originally
introduced to model transmission of information (electrical signals in Heaviside's theory), and not just a mathematical trick to get a Hamiltonian dilation of resistive models. The coupling to transmission lines can give a model of a linear passive heat bath, but they can also serve to relay (noisy!) signals into a system.

An obvious question to ask at present is the following: What are the connection rules for these non-Markovian quantum circuits corresponding to the series product, feedback reduction, etc.? At one level the answer should be the same as for classical circuits: e.g., Kirchhoff's current/voltage laws. However, one can consider non-linear models where the terminal charge (or current) couples to a quantum mechanical system. This type of modeling is generic in the \emph{SLH} framework for quantum feedback networks \cite{GouJam09a,GouJam09b} and here one has modular rules for building networks from Markovian components. 

To answer the question, we need to look at the network structure of systems connected by non-Markovian transmission lines. While the general theory will be modular, it will be much less applicable than the Markovian quantum feedback network theory, but there are important consistency questions arising. For instance whether the connection rules and the Markov approximation are commuting steps. We shall address these in later publications. However, in this publication we have set out the theory for simple linear components and derived their stability criteria, and their Markov approximation.

\end{document}